\documentclass[11pt, a4paper]{amsart} 
\def\N{{\rm I\kern-0.16em N}}
\def\R{{\rm I\kern-0.16em R}} 
\def\E{{\rm I\kern-0.16em E}}
\def\P{{\rm I\kern-0.16em P}} 
\def\F{{\rm I\kern-0.16em F}}
\def\B{{\rm I\kern-0.16em B}} 
\def\Q{{\rm I\kern-0.16em Q}} 
\def\C{{\rm I\kern-0.16em C}}

\font\eka=cmex10
\def\ind{\mathrel{\hbox{\rlap{%
\hbox to 7.5pt{\hrulefill}}\raise6.6pt\hbox{\eka\char'167}}}}
\def\L {{\rm I\kern-0.16em L}}
\newtheorem{example}{Example}[section]
\newtheorem{theorem}{Theorem}[section]
\newtheorem{remark}{Remark}[section]
\newtheorem{lemma}{Lemma}[section]
\newtheorem{corollary}{Corollary}[section]

\newtheorem{proposition}{Proposition}[section]
\newtheorem{definition}{Definition}[section]
\numberwithin{equation}{section}
\usepackage{amssymb}
\usepackage{amsmath}
\usepackage{verbatim}
\usepackage{graphicx,color}
\begin{document}

\markboth{Gasbarra, Morlanes, Valkeila}{Initial enlargement in a Markov chain market model}
\date{\today}
%%%%%%%%%%%%%%%%%%% Publisher's Area please ignore %%%%%%%%%%%%%%%%%%%%%%%
%\catchline{}{}{}{}{}
%%%%%%%%%%%%%%%%%%%%%%%%%%%%%%%%%%%%%%%%%%%%%%%%%%%%%%%%%%%%%%%%%%%%%%%%%%
\title{Initial enlargement in a Markov chain market model}
\author{Dario Gasbarra}
\address{University of Jyv\"askyl\"a, Department of Mathematics and Statistics\\
PL 35 (MaD) 40014 University of Jyv\"askyl\"a\\
Finland\\
dario.gasbarra@jyu.fi}
\author{Jos\'e Igor Morlanes}
\address{ Aalto University, Department  of Mathematics and Systems Analysis\\
P.O. Box 11100, FI-00076 Aalto\\  
Finland\\
Igor.Morlanes@tkk.fi}
\author{Esko Valkeila}
\address{Aalto University, Department  of Mathematics and Systems Analysis\\
P.O. Box 11100,  FI-00076 Aalto\\
Finland\\
Esko.Valkeila@tkk.fi}

\maketitle

%\begin{history}
%\received{(Day Month Year)}
%\revised{(Day Month Year)}
%\accepted{(Day Month Year)}
%\comby{(xxxxxxxxxx)}
%\end{history}

\begin{abstract}
Enlargement of filtrations is a classical topic in the general theory of stochastic processes.
This theory has been applied to stochastic finance in order to analyze models with insider information.
In this paper we study initial enlargement in a Markov chain market model, introduced by R. Norberg. 
In the enlarged filtration several things can happen: some of the jumps times  can 
be accessible or predictable, but in the original filtration all the jumps times are totally inaccessible.
But even if the jumps times change to accessible or predictable, the insider does not necessarily have arbitrage 
possibilities.
\end{abstract}


\section{Introduction}  
Enlargement of filtrations is a classical topic in the general theory of stochastic processes
\cite{Jeulin1980}. 
This theory has been applied to stochastic finance in order to analyze models with insider information
(see for example \cite{Ankirchner2005,Ankirchner2008,adi2006,Gasbarra2006}).
In this paper we study initial enlargement in a Markov chain market model, introduced by R. Norberg
\cite{Norberg2005}.  In this model the state of economy is modeled by a finite state Markov chain, 
and the state of economy determines the dynamics for the risky assets. 

The ordinary agent has the information described by the filtration generated by 
an observable  process, but the insider has the additional information given by  a certain random variable.

We assume that the ordinary agent has no arbitrage possibilities. 
Then, in the initial enlargement  the following things can happen; 
\begin{itemize}
\item In the original filtration the jump times are totally inaccessible, but in the enlarged filtration 
there can be accessible and predictable jump times.
\item Independently of the possible changes in the properties of jump times, the insider may have arbitrage 
possibilities, or may not have arbitrage possibilities.
\end{itemize}

The motivation for this study comes from the jump model example introduced by
A. Kohatsu-Higa \cite{Kohatsu-Higa2007}.  Our results show some additional features in the enlargement theory
for processes with jumps. 
\section{Markov Chain Market Model}
\subsection{States of the economy}
We describe the model introduced by  R. Norberg. 
We work with  probability space $(\Omega, \mathcal{F},\mathbb{P})$. 

The state of the economy is given by a  process $Y$. Next we list the 
properties of the process $Y$: $Y=(Y_t)_{t\geq0}$ is a time-homogeneous Markov process with 
finite state space $\mathcal{Y}=\{1,\cdots,n\}$ and the paths of $Y$ are right continuous
with left-hand limits.

We denote the transition probabilities by 
$P_t^{ek} = \mathbb{P}(Y_{t+s}=k\big\vert Y_s=e) $, $s\ge 0$,      ${P}_t=\{P^{ek}_t\} $ 
is the transition matrix and  $\mathbf{\Lambda}=\{\lambda^{ek}\}$
is the intensity matrix.

The states of the Markov chain determine the dynamics of the risky assets.
\subsection{Dynamics of risky assets}
The market model has $m+1 $ assets $S = (S^0, S^1 , \dots , S^m)$. We describe their dynamics 
with the help of the state process $Y$.

The counting process 
$$
N_t^{ek}=\#\{\tau: 0<\tau\leq t, Y_{\tau-}=e , Y_{\tau}=k\}
$$
counts direct transitions of $Y$ from state $e$ to state $k$ during the time interval $(0,t]$.

The bank account $S^0 _t $ has dynamics
$$
S^0_t = \exp\left( \int^t_0 r_u du \right) = \exp\left( \sum_e  \int^t_0 r^e \textbf{1}_{\{Y_{u-}=e\}} du\right),
$$
where $r_t = r^{Y_t}$, i.e. the short rate depends on the state of the economy. 

The rest of the assets have  the following dynamics
$$
S^i_t = S^i_0\exp \left(\sum_e \biggl\{
 \int^t_0  \mu^{ie}\textbf{1}_{\{Y_{u-}=e\}} du + \sum_{k}\beta^{iek} N_t^{ek}\biggr\} \right) ,
$$
where $\mu^{ie}\in \mathbb R , \beta ^{iek}\in \mathbb R$. 
Then the logarithmic discounted prices $L^i  = \log ({S^i }/{S^0}) $ have dynamics 
$$
dL_t^i =\sum_e (\mu^{ie}-r^e) \textbf{1}_{\{Y_{t-}=e\}}dt + \sum_e \sum_{k}\beta^{iek}dN_t^{ek} .
$$

Later we shall work with three different filtrations: with the filtration generated by the 
Markov chain process $Y$, which we denote by $\mathbb F$, or by $\mathbb F ^Y $,  and with 
two initially enlarged filtrations, which we denote by $\mathbb G $ and $\mathbb G ^{ H }$. 
We will specify the filtrations $\mathbb G$ and $\mathbb G ^{H}$ later. 
Note that the filtration $\mathbb F$ is also generated by the counting processes $N^{ek}$, 
$e,k\in \mathcal Y $ with $e\neq k$.

\subsection{No-arbitrage criterion in the Norberg model}
From the definition of the model we have that  
\begin{eqnarray*}  M^{ek}_t= 
 N^{ek}_{t}-\int_0^t\lambda^{ek}\mathbf{1}( Y_s=e)  ds,
\end{eqnarray*}
are  mutually orthogonal $(\mathbb{F},\mathbb{P})$-martingales: 
indeed we have for $e,k,l,p\in \mathcal Y $
when $k\neq p$ or $e\neq l$ that $[M^{ek},M^{lp}]_T=\sum _{s\le T} \Delta M^{ek}_s \Delta M^{lp}_s = 0$ a.s., 
and this implies mutual orthogonality 
in the sense of \cite[Definition VI.6.12]{Chinese}.

\begin{definition}\label{d:equivalence} The  intensity matrices
 $\mathbf{\Lambda}=(\lambda^{ef} ) $ and $\mathbf{\widetilde{\Lambda}}= (\widetilde \lambda^{ef} )  $ 
are equivalent when $\forall e,f\in  {\mathcal Y }$,
$  \lambda^{ef} > 0  \Longleftrightarrow   \widetilde \lambda^{ef} > 0 $.
\end{definition} 

In order to make the discounted stock price process a martingale
we should have that the new intensity $\widetilde \lambda ^{ek}$ satisfies 
for all $e \in \mathcal{Y}$ :
$$
\mu^{ie}-r^e = -\sum_{k\in \mathcal{Y}^e}\gamma^{iek}\widetilde{\lambda}^{ek},\qquad  i=1\dots m ,
$$ 
where $\mathcal{Y}^e=\{k: \lambda^{ek}>0\}$ is the set of states directly reachable from state $e$,
and $\gamma ^{ek} = e^{\beta ^{ek}}- 1 $. 
Rewrite this in matrix form as 
$$
\leqno(\mbox{NA}) \ \ \ \ \  
r^e {\bf 1} - \mu ^e = \Gamma ^e \widetilde \lambda ^e ,
$$
where $e=1,\ldots,n$, $\mathbf{1}$ and $\mu^e=\left( \mu^{ie}\right)_{i=1,\ldots,m}$ are $1\times m$ row vectors,

$\mathbf{\Gamma}^e=\left( \gamma^{ief}\right)_{i=1,\ldots,m}^{f\in\mathcal{Y}^e} $,  
$\widetilde{\mathbf{\lambda}}^e=\left( \widetilde{\lambda}^{ef}\right)_{f\in\mathcal{Y}^e}$.
We can now summarize the situation:
\begin{proposition}\cite{Norberg2005}
 Assume that we can find $\widetilde \Lambda $, equivalent to $\Lambda$, such that 
(NA) holds, then defining  $\mathbb Q$ by $d  \mathbb Q _t = Z_t d\mathbb P _t $ with the density 
\begin{eqnarray*}
Z_t &= &  \exp\biggl(  \sum_{e\in \mathcal{Y}}  
\sum_{k\in \mathcal{Y}^e} \bigl( \log( \widetilde  \lambda^{ek} )-
\log (\lambda^{ek} )  \bigr) N_t^{ek} \biggr) \\
& & \times \exp \biggl(
 \sum_{e\in \mathcal{Y}} \sum_{k\in \mathcal{Y}^e} 
  \int_0^t(\lambda^{ek}-\widetilde  \lambda^{ek}) \textbf{1}(Y_{s-}=e) ds \biggr) .
\end{eqnarray*}
we obtain a martingale measure for the Norberg market model.
\end{proposition}

Without loss of generality we can assume that the state dependent interest rate $r^e= 0$.

We give two basic examples, which we use to illustrate various aspects of the Norberg model
in connection to initial enlargement.

\begin{example}\label{ex:one}
In \cite{Kohatsu-Higa2007}, Kohatsu-Higa introduced the following model for the stock price:
$$
L_T = \log (S_T)=L_0 + \beta ^+ N^+_T + \beta ^- N^- ;
$$
here $\beta ^+ > 0 , \beta ^- < 0 $, $N^+ $ and $N^- $ are  Poisson process with respective intensities $\lambda ^+ $
and $\lambda ^- $, counting respectively the upward and downward jumps, respectively.

This can be put  in the Norberg model as follows. The state space is  
$\mathcal Y = \{ 1,2,3 \}$, and there  is  one stock $S$.

The parameters are $\beta ^+>0$ $\beta ^-<0$, and the drift $\mu = \mu ^i $, for $i=1,2,3$. Take 
\begin{eqnarray*}
\frac{ dS_t }{ S_t} = \gamma^+ ( dN^{12}_t + dN^{23}_t + dN^{31}_t ) +\gamma^- ( dN^{13}_t + dN^{21}_t + dN^{32}_t ) + \mu dt 
\end{eqnarray*}
where $\gamma^{\pm} =\bigl(\exp(\beta^{\pm})-1\bigr)$ and
\begin{eqnarray*}
 dN_t^{ij} - {\bf 1}( Y_{t-} = i  ) \lambda^{ij} dt 
\end{eqnarray*}
are martingale increments for $i\ne j$ under the  measure $P$ with $\lambda^{ij} >0$.

Take now $\lambda^{12}=\lambda^{23}= \lambda ^{31}=\lambda ^+ $ and 
$\lambda^{21}=\lambda ^{32} = \lambda ^{13}=\lambda ^- $. Then
the aggregated processes
$N^+_t = (N^{12}_t+N^{21}_t+N^{31}_t)$ and $N^-_t = ( N^{13}_t+N^{21}_t+N^{32}_t)$ 
have  deterministic compensators 
\begin{eqnarray*} 
  \sum_{i=1}^3 \int_0^t  \lambda^{\pm}{\bf 1} (Y_{s-} = i) ds = \lambda^{\pm} t
\end{eqnarray*}
and since $[N^+ ,N^-]=0$ a.s., by Watanabe's characterization $N^+,N^-$ are
independent Poisson processes (see \cite{Jacod-Shiryaev2003}).

To check the (NA) condition we find $\widetilde\lambda^{\pm} >0 $
such that
\begin{eqnarray*} 
  \gamma^{+ } \widetilde \lambda^{+} + \gamma^{- } \widetilde \lambda^{-}  +\mu =0 
\end{eqnarray*}
and then we obtain an equivalent risk-neutral measure  $\mathbb Q$  with intensities
 $\widetilde \lambda^{12 } =\widetilde \lambda^{23 } = \widetilde \lambda^{31 } = \widetilde \lambda^+$,
and $\widetilde \lambda^{21 } =\widetilde \lambda^{32 } = \widetilde \lambda^{13 } = \widetilde \lambda^-$.

The model is incomplete: 
\begin{equation} \label{eq:n-a}
\widetilde\lambda^{-}  =    -\biggl(
 \mu+ \gamma^{+} \widetilde \lambda^{+} \biggr) \bigg / \gamma^->0
\end{equation}
is a solution
for any fixed $\mu$ and  large enough $\widetilde\lambda^{+}>0$, since $\gamma^{-}<0$
and $\gamma^+>0$. Hence there are many martingale measures.

\end{example}

\begin{example}\label{ex:two}
Later we will illustrate what can happen in the initial enlargement using the following model.
\begin{itemize} 
\item The economy can be in two different states: $\mathcal Y = \{1,2 \}$.
\item Write $\mu^- = \mu ^1 $, $\mu ^+ = \mu ^2 $, $\lambda ^+ = \lambda ^{12} $, $\lambda^{-} = \lambda ^{21}$ and similarly with $\gamma ^+ $, $\gamma ^-$,
$N^+ $ and $N^-$.
\item Assume that $\mu^+ > 0$, $\gamma ^+ > 0$, $\mu ^- < 0 $ and $\gamma ^- <  0 $. 
\item We have only one stock and 
$$
L_T = \log (S_T) =L_0 + \beta  ^+ N_T^+ + \beta ^-N^-_T + \int _0^T \mu ^ {Y_u}du .
$$  
\end{itemize}
It is easy to see that the (NA) condition holds for the assumed parameter values.

Note that here, in contrast to Example \ref{ex:one}, the processes $N^+ $ and $N^-$ are not independent.
\end{example}

\subsection{An alternative description of the model}
\label{ss:scenario}
The randomness of the model comes from the finite state Markov processes $Y$. Alternatively,
we can consider the matrix valued counting process $\mathbf N = (N^{el})_{e,l\in \mathcal Y , e\neq l}$,
where $N^{el}$ counts the direct transitions from state $e $ to state $l$. 

The other possibility is  
to consider a single counting process $N$, where $N = \sum _{e,l}N^{el}$, and keep track, 
how the Markov process $Y$ behaves at the jump times of $N$. More precisely, 
this information is given by the \emph{scenarios}.
A scenario  $h=( n; e_0,e_1,\dots, e_{n} )$ gives information about the associated Markov chain, 
here $n\ge 0 $ is the number of changes 
in the economy, $e_0$ is the initial state, and
 and $e_i\neq e_{i+1}$, $i=1, \dots , n-1$, are the states of the economy in the scenario $h$. 
For example, $(0;e_0)$ is the scenario, 
where  there are no changes in the economy. Notation: $e^{0:n} = e_0,e_1,\dots, e_{n} $.
The random scenario $H_T = (N_T ; Y_0, Y_{\tau _1}, \dots , Y_{N_T}) $, where 
$\tau _k $ is the k$^{th}$- jump time of the counting process $N$, together with the aggregated 
counting process $N$ defined above, has the same information as the matrix valued counting 
process  $\mathbf N$. The scenarios  will be useful for us both in some computations and 
in the analysis of the initial enlargement.

\section{Calculation of the insider's compensator: Classical theory}\label{insider dynamics}

\subsection{A martingale representation result}

Let $(\Omega,{\mathcal F},\mathbb{P},\mathbb{F})$ be a filtered probability space 
with $\mathbb F=(F_t)_{t\geq 0}$ any filtration, not necessarily the 
filtration generated by the counting process $\mathbf N $. 
Next we study how the the compensator $\Lambda ^{ek}$ of the counting process 
$N^{ek}$ is computed in the initially enlarged filtration, $e,k\in \mathcal Y , \ e\neq k$. 
We shall simply write $N$ and $\Lambda $, instead of $N^{ek}$ and $\Lambda ^{ek}$. 

So, assume that $N=(N_t)_{t\geq 0}$ is a $\mathbb{F}$-adapted counting process with 
$\mathbb F$-compensator $\Lambda=(\Lambda_t) _{t\geq 0}$.
We consider, in the next subsection, an enlargement of the filtration $\mathbb{F}$ by a 
random variable $\vartheta$.
To be able to compute the compensator of the process $N$ in this enlarged filtration, 
we need a few results given below.

We use the notation ${}^p X$ (respectively ${}^o X$) to be the predictable (respectively optional) 
projection of $X$ and $X^p$ the dual predictable projection (respectively $X^o$) w.r.t 
$(\mathbb{F}, \mathbb{P})$, unless otherwise stated. 

The next lemma is a version of the martingale representation theorem in our context. 

\begin{lemma}\label{MRT}
Let $N$ be a counting process with continuous compensator $\Lambda=N^p$ w.r.t. the filtration $\mathbb F$. Denote 
by $\widetilde N = ( N-\Lambda)$  the compensated process. Then every $\mathbb R^d$ valued $\mathbb F$-local martingale $(M_t)_{t\geq 0}$ has the representation 
\begin{eqnarray*} 
 M_t=M_0 + \int_0^t (\widehat M_s - M_{s-}) d\widetilde N_s + U_t
\end{eqnarray*} 
where $\widehat M$ is $\mathbb F$-predictable,
and $(U_t)_{t\geq 0}$ is a $\mathbb F$-local martingale with  $\langle \widetilde N, U \rangle =0$.
\end{lemma}

\begin{proof} 
% For a proof check \cite[p.313]{Shiryaev1974}.
The proof is essentially a modification  of the results in 
\cite[Theorem III.4.20 and Lemma III.4.24]{Jacod-Shiryaev2003}, 
we give it here  to clarify the nature of  $\widehat M$.

Put $\mathbb R ^d_0 = \mathbb R^d\setminus \{0\}$. 
Let $\mu = \mu ^{M,N}$ be the jump measure of the process $(M,N)$. 
Note that $\mu$ is an integer valued random measure  on $E\times \mathbb R^+$, where   
$E = \mathbb R^d_0 \times \{0,1\}$ equipped with  the Borel $\sigma$-algebra $\mathbb E $.

Obviously $\mu^M (\cdot,dt ) = \mu (\cdot \times \{0,1\},dt)$ and 
$\mu ^N(dt)  = \mu (\mathbb R^d_0 \times \{ 1\},dt) $. Let 
$\nu = \nu ^{M,N}$ be the $(\mathbb F , \mathbb P)$ compensator of $\mu$. 

The uniqueness of the compensator implies that 
$\nu ^M (\cdot ,dt ) = \nu (\cdot \times \{ 0,1\} ,dt) $ and $\Lambda(dt) = \nu ^N(dt) = \nu (\mathbb R^d_0 \times \{1\},dt ) $. We introduce the $\mathbb R^d$- valued process 
\begin{eqnarray*}
A_t := \int _0^t \int_{\mathbb R ^d _0 } x \nu ( dx\times \{1\}, ds ) ,\end{eqnarray*}
which is predictable with finite variation. The corresponding $\mathbb R^d$-valued
predictable random measure $A(dt)$ satisfies $A \ll \Lambda  $ on ${\mathcal B}(\mathbb R^+)$,
with R-N derivative $\rho  := \frac{dA}{d\Lambda } $. Define 
\begin{eqnarray*}
\widetilde M _t =\int _0^t \rho _s d \widetilde N_s\; ,
\end{eqnarray*}
where $\widetilde N = N - \Lambda $.

The martingale $M$ has a decomposition
$$
M_t = M_0 + M^c_t + \int _0^t \int _{\mathbb R^d_0} x (\mu ^M - \nu ^M )(dx,ds) ;
$$
Put $U = M - \widetilde M$. Then $U$ has a decomposition 
$$
U_t = M_0 + M ^c _t + \int _0^t\int _ {\mathbb R ^d _0 } x (\mu ^M - \nu ^M )(dx,ds ) - \widetilde M _t .
$$
We will show that $\langle U , \widetilde N \rangle = 0 $. 

Obviously $[M^c , \widetilde N ] = 0 $ and so $\langle M^c , \widetilde N \rangle = 0 $.
Next, we have 
\begin{eqnarray*} 
 \int _0^t \int _{\mathbb R ^d _0 } x (\mu ^M - \nu ^M )(dx,ds ) & =&  \int _0^t \int _{\mathbb R^d _0 } x (\mu - \nu )
(dx\times\{0,1\} ,ds )  \\
= \int _0^t \int_{ \mathbb R ^d _0} x (\mu - \nu )(dx\times \{0\} ,ds )  &+ &
\int _0^t \int_{ \mathbb R ^d _0} x (\mu - \nu )(dx\times \{1\} ,ds ) \\ &  =: & M^{d,0}_t + M^{d,1}_t .
\end{eqnarray*}
By construction, $[M^{d,0}, \widetilde N ] = 0 $, and hence $\langle M^{d,0}, \widetilde N \rangle=0$. 
Finally,
\begin{eqnarray*}
\langle 
M^{d,1},\widetilde N\rangle _t  = \int _0^t \int _{\mathbb R^d_0} x  \nu (dx\times\{ 1\}, ds) = \int _0^t \rho_s ds ,
\end{eqnarray*}
and this gives $\langle M^{d,1} - \widetilde M, \widetilde N \rangle = 0 $.  Finally, put $\widehat M = M_- + \rho$, and we have the claim.

\end{proof}

\begin{remark}\label{note1}
Note that we have the interpretations
\begin{eqnarray*} 
 d\Lambda _t = \mathbb P(\Delta N_t= 1 | F_{t-}) , \quad dA_t = 
E_{\mathbb P} ( \Delta M_t \Delta N_t | F_{t-}) \; .
\end{eqnarray*}
%$d\Lambda _t $ has the following interpretation:\newline
% $d\Lambda _t = P(\Delta N_t= 1 | F_{t-}) $,  and similarly $dA_t $ has the interpretation \newline
%$dA_t = E ( \Delta M_t \Delta N_t | F_{t-})$.
 Hence $\rho _t$ has the interpretation 
$$
\rho _t = \frac{E _{\mathbb P}( \Delta M_t {\mathbf{1}}( \Delta N_t=1) | F_{t-})}
{\mathbb P(\Delta N_t= 1 | F_{t-})} =: E_{\mathbb P}(\Delta M_t | F_{t-}, \Delta N_t = 1 ) ,
$$
which is well-defined, since $\rho$ is a Radon-Nikodym derivative.  
Similarly, $\widehat M $ has the interpretation 
$$
\widehat M_t = E_{\mathbb P}(  M_t \big\vert F_{t-}, \Delta N_t = 1 ) .
$$
\end{remark} 
We recall the definitions of the stopped 
$\sigma$-algebras (\cite[I.1.1b]{Jacod-Shiryaev2003}) associated to
a  stopping time $\tau$:
\begin{eqnarray*} &&
{ F}_{\tau} := \bigl \{ A\in { F}:
  A  \cap \{ \tau \le t \} \in {F}_t \quad \forall t \ge 0 \bigr
\} 
\\ &&
{ F}_{\tau-} := \sigma\bigl(  A  \cap \{ t < \tau  \} : t\ge 0, A\in {F}_t
\bigr)
\end{eqnarray*}
It follows that ${ F}_{\tau-} \subseteq { F}_{\tau}$, and 
by taking $A=\Omega$ in the definition,  $\tau$ itself is $F_{\tau-}$-measurable.
In simple words, ${ F}_{\tau-}$ contains the information
about $\tau$  and everything that happened before it,
while ${ F}_{\tau}$  contains also the information
which comes  with $\tau$.

Since the simple left-continuous  adapted processes 
\begin{eqnarray*}
 K_t(\omega) = {\bf 1}_A( \omega) {\bf 1}(  u< t ) ,\quad  u\ge 0, \; A\in F_u
\end{eqnarray*}
generate the predictable $\sigma$-algebra,  it follows that
\begin{eqnarray*}
{ F}_{\tau-}=
 \sigma\bigl ( K_{\tau}{\bf 1}(\tau < \infty): K \mbox{ is $\mathbb F$-predictable } \bigr)
\end{eqnarray*}

\begin{lemma} \label{equivalence:lemma}
  Let $\tau$ be a stopping time in a right-continuous filtration $\mathbb F=(F_t)$
completed by the $\mathbb P$-null sets. Denote $n_t={\bf 1}(\tau\le t)$ and
$\tilde n= (n-\varrho)$ with $\varrho= n^p$. Then for all $\mathbb F$-martingales $u$ 
\begin{eqnarray*} 
\langle  \tilde n,u \rangle=0 \Longrightarrow [ \tilde n,u ]=[n,u]=0 \; ,
\end{eqnarray*} 
if and only if ${ F}_{\tau}= { F}_{\tau-}$.
\end{lemma}

\begin{proof} 
Let $u_t$ be a $\mathbb F$-martingale
with $\langle u, \tilde n\rangle=0$.

The r.v. $\Delta [n,u]_{\tau}= \Delta u_{\tau}$ is 
${ F}_{\tau}$-measurable. By the assumption ${F}_{\tau}={ F}_{\tau-}$,
there  is a predictable process $k_t$
such that ${\mathbf 1}(\tau<\infty) \Delta u_{\tau} ={\mathbf 1}(\tau<\infty) k_{\tau}$, which means
\begin{eqnarray*}
[u,n]_t =\int_0^t k_s dn_s =( k\cdot n)_t \ .
\end{eqnarray*}
In the notation of Lemma \ref{MRT} $k_t =( \widehat u_t -u_{t-})$. Note that
\begin{eqnarray*}
  [u,\tilde n] = [u, n] -[u,\varrho]
=  (k\cdot n) - [u,\varrho] = (k\cdot \varrho) +
(k\cdot \tilde n)  - [u,\varrho]
\end{eqnarray*}
is  a local martingale since by assumption $\langle u, \tilde n \rangle=0$.
Since
\begin{eqnarray} \label{eq:loc.martingale}
[u,\varrho]_t = \int_0^t \Delta \varrho_s du_s
\end{eqnarray}
is also  a  local martingale, the predictable process $(k\cdot \varrho)$ is 
 a local martingale with finite variation, therefore
\begin{eqnarray*}
 (k\cdot \varrho )=0 \quad \mbox{  and  hence also } \quad ( k\cdot n)=[u,n]=
[u,\tilde n]=0.
\end{eqnarray*}

Next we show that if
\begin{eqnarray*}
 [  u ,  n  ]= 0  \mbox{ for all $\mathbb F$-martingales $u$
with  $\langle  u , \tilde n \rangle =0  $ },
\end{eqnarray*}
then necessarily ${ F}_{\tau-} = { F}_{\tau}$.
If this is not the case, there is   $A\in ( { F}_{\tau} \setminus { F}_{\tau-} )$ with $\mathbb P(A)>0$, 
and we find a bounded and ${ F}_{\tau}$-measurable random variable
\begin{eqnarray*}
X(\omega) : =
 {\bf 1}_A(\omega)-\mathbb P( A | {\mathcal F}_{\tau-} )(\omega)  \not\equiv 0 
\end{eqnarray*}
with  $E_{\mathbb P} ( X | {\mathcal F}_{\tau-} )(\omega)=0$. We show first that
$u_t(\omega):=X(\omega)n_t(\omega)$ is a $\mathbb F$-martingale:
\begin{itemize} 
 \item 
$u_t$ is $\mathbb F$-adapted since $X$ is ${ F}_{\tau}$-measurable.
\item
For $s \le  t$ and $A\in F_s$, $(n_t-n_s){\bf 1}_A=n_t ( 1-n_s){\bf 1}_A$
is  ${ F}_{\tau-}$-measurable,
 since $(1-n_s) {\bf 1}_A$ is ${ F}_{\tau-}$-measurable by definition
and $\tau$ is ${ F}_{\tau-}$-measurable.
The martingale property follows:
\begin{eqnarray*}
E_{\mathbb P}( (u_t -u_s)  {\bf 1}_A  ) = 
 E_{\mathbb P}\bigl( X (n_t -n_s) {\bf 1}_A \bigr) = 
E_{\mathbb P}\bigl(  E_{\mathbb P}( X  \vert {\mathcal F}_{\tau-} )  (n_t-n_s) {\bf 1}_A \bigr ) =0
\end{eqnarray*}
\end{itemize}
Note also that 
\begin{eqnarray*}
0 \not\equiv u_t = [u,n]_t = [u,\tilde n]_t + [u,\varrho ]_t
\end{eqnarray*}
where  $[u,\varrho]$ is a local martingale by (\ref{eq:loc.martingale}).
We see that  $[u,\tilde n]$ is a local martingale 
which implies $\langle u,\tilde n \rangle= 0$.
\end{proof}

\textbf{Assumption 1 }  The   jump times $\tau_k$ of $N$
 satisfy ${F}_{\tau_k}={F}_{\tau_k-}$, with continuous $\mathbb F$-compensator.

\begin{corollary} \label{c:mrt} Under assumption 1,
let $\eta(\omega)$ be  a ${\mathcal F}$-measurable  $\mathbb R ^d$- valued random variable. 
If $f$ is a bounded measurable function, then the optional projection ${}^o f(\eta ) $ of $f(\eta )$
is a $\mathbb F$- martingale. Hence ${}^of(\eta ) $ has a representation  
\begin{eqnarray*} 
i) &&   {}^o f(\eta)_t = E_{\mathbb P} (f(\eta))+ \int_0^t ( \widehat{  {}^o f(\eta)} _s -{ }^o f(\eta)_{s-}) d\widetilde N_s + U_t(f) ,\\ && 
\mbox{   $(U_t(f))$ is a $\mathbb F$-martingale with } 
[ \widetilde N,
U(f )] =[ N,U(f) ]= 0,
\\  
ii) 
&&   {}^p f(\eta)_{s} = {} ^o f(\eta)_{s-}
\end{eqnarray*} 

\end{corollary}

\subsection{Compensator after initial enlargement}

To compute the compensator of $N$ in the initially enlarged filtration $\mathbb G$, where 
$G_t = \cap _{u> t} (F_t \vee \sigma(\vartheta))$, we used the approach initiated in  \cite{Gasbarra2006}, 
and developed further in \cite{Ankirchner2005}:

Consider the measurable product space $(\Omega \times \mathbb R^m,\mathcal{F}\otimes\mathcal{B}(\mathbb R^m))$ denoted by $(\bar{\Omega},\bar{\mathcal{F}})$. Define the map 
\begin{equation*}
\begin{array}{ccccc}
\Phi:&(\Omega,\mathcal{F})& \longrightarrow& (\bar{\Omega}, \bar{\mathcal{F}})\\ &\omega& \longmapsto& (\omega, \vartheta(\omega))\\ 
\end{array}
\end{equation*}

We denote by $\bar{\mathbb{P}}$ the image of the measure $\mathbb P$ under $\Phi$, i.e. $\bar{\mathbb{P}} = \mathbb{P}_\Phi $.
Endow the space $\bar{\Omega}$ with the $\bar{\mathbb P}$-completed filtration $\bar{\mathbb{F}}=(\bar{F}_t)_{t\geq 0}$ where 
\begin{equation*}
 \bar{F}_t=\bigcap\limits_{u>t}( F_u\otimes \mathcal{B}(\mathbb R^m)) \vee {\overline N} ,\quad
 {\overline N}= \bigl \{ \bar A \subseteq \bar{\Omega}: \bar{\mathbb{P}} (\bar A )=0 \bigr\} \; .
\end{equation*}
%{\color{red} THE REFEREE ASKS WHETHER WE NEED TO COMPLETE
% THE FILTRATION $\bar{\mathbb{F}} $ WITH RESPECT TO
%THE MEASURE $\bar{\mathbb{P}}$ . DO WE NEED TO DO IT ?
%WE HAVE ASSUMED THAT 
%${\mathbb{F}} $ is ${\mathbb{P}}$-COMPLETE }
We will  consider the initially enlarged filtration $\mathbb{G}=(G_t)_{t \in [0,T]}$ with 
$G_t =\bigcap\limits_{u>t}(F_u \vee \sigma( \vartheta))$,
where $\vartheta \in {\mathcal L}^0(\Omega,{\mathcal F},\mathbb{P},\mathbb{F})$ is a $m$-dimensional random variable.

Consider also  the filtered spaces 
$$
\left(\bar{\Omega}\times \mathbb R_+ ,\, \bar{\mathbb{F}}\otimes \mathcal{B}\left(\mathbb R_+\right),\,\bar{\mathbb{P}}\right)
$$
and 
$$
\left(\Omega\times \mathbb R_+ ,\, \mathbb{G}\otimes \mathcal{B}\left(\mathbb R_+\right),\, \mathbb{P} \right) .
$$ 

Recall the following facts from \cite{Ankirchner2005}; let  $\bar X$ be a stochastic process defined on 
$(\bar \Omega\times \mathbb R_+ , \bar{\mathbb{F}} \otimes \mathcal B (\mathbb R ^+ ) )$:
\begin{itemize}
\item If $\bar X $ is $\bar{\mathbb F}$- predictable
(resp. $\bar{\mathbb F}$- optional), then $X = \bar X \circ \Phi $ is $\mathbb{G} $- predictable
(resp. ${\mathbb G}$- optional).
\item If $\bar \tau $ is a $\bar{\mathbb F}$- stopping time, then $\tau = \bar\tau\circ \Phi$ is a 
$\mathbb{G} $- stopping time.
\item If $\bar M$ is $(\bar{\mathbb F}, \bar{\mathbb{P}} )$- local martingale, then $M = \bar M \circ \Phi $ is a 
$(\mathbb{G},\mathbb{P} )$- local martingale. 
\end{itemize}

For example, let $\bar X$ be a simple $\bar{\mathbb F}$-predictable process: 
\begin{eqnarray*}
\bar{X}(\omega,\ell,u)=\textbf{1}_A(\omega)\textbf{1}_B(\ell) {\bf 1}_{(s,t]}(u) ,
\end{eqnarray*}
where $s < t \le T$, $A\in F_s $ and $B \in \mathcal B (\mathbb R^m) $.  Then 
$X = \bar X \circ \Phi $ is 

\begin{equation}
 \label{processinomega}
X = \bar X \circ \Phi = \textbf{1}_A \textbf{1}_B(\vartheta ) 1_{(s,t]}
\end{equation}

$\mathbb G $-predictable process.

Extend $N$ to $\bar\Omega \times \mathbb R_+ $ by $\bar N$, where $\bar{N}(\omega,\ell,u)=N(\omega,u)$. 

Let $\bar \pi $ be the measure generated by $\bar N $ on $(\bar F \otimes \mathcal B(\mathbb R_+))$:

\begin{eqnarray*} 
\bar\pi(\bar{Y}):=E_{\bar{\mathbb{P}}}\left(\int_0^\infty \bar{Y}_u d\bar{N}_u  \right)
=E_{\mathbb P}( \mathbf{1}_C \mathbf{1}_B(\vartheta) (N_t -N_s)), 
\end{eqnarray*} 
where 
$$
\bar{Y}(\omega , \ell, u) = \mathbf{1}_C(\omega ) \mathbf{1}_B (\ell) \mathbf{1}_{(s,t]}(u) 
$$

Since $\bar N$ is optional 
with respect to the history $\bar{\mathbb{F}}$, the measure $\bar \pi $ is also optional:
for any bounded non-negative $\mathcal{F}\otimes \mathcal B (\mathbb R ^m) \otimes \mathcal B (\mathbb R_+ )$ measurable 
process $\bar Y $ we have
$$
\bar \pi (\bar Y ) = \bar \pi ( {}^{\bar {\mathbb F}, o }\bar Y ) 
$$
(see \cite[p. 141]{Chinese} for more details). 

Denote by $\pi $ the measure generated by $N$
on $\mathbb{F}\otimes \mathcal{B}\left(\mathbb R_+\right)$. Then for optional $\bar Y$, put 
 $Y = \bar Y \circ \Phi$, and we have $\pi (Y) = \pi ({}^o Y) $ (see \cite{Ankirchner2005}). 
Apply this to $X$ of the form 
(\ref{processinomega}) with $A \in F_s $, and we get 
\begin{eqnarray*}
\bar\pi( \bar X) = \pi (X) = \pi ({}^{o}X)
=E_{\mathbb P}\bigl( \mathbf{1}_A \bigl\{ 
 {}^o (\mathbf{1}_B( \vartheta))_t N_t - {}^o (\mathbf{1}_B(\vartheta))_s N_s\bigr \} \bigr ).
\end{eqnarray*}
We can now continue using Corollary \ref{c:mrt}, the continuity of $\Lambda $, lemma \ref{equivalence:lemma} under
 assumption 1 
 and integration by parts to obtain

\begin{eqnarray*} 
&&E_{\mathbb P}\bigl( \mathbf{1}_A \bigl\{ 
 {}^o (\mathbf{1}_B( \vartheta))_t N_t - {}^o (\mathbf{1}_B(\vartheta))_s N_s\bigr \} \bigr )
= \\ &&
E_{\mathbb P}\biggl( \mathbf{1}_A \int_s^t {}^ o(\mathbf{1}_B( \vartheta))_{u-}  dN_u\biggr) + 
E_{\mathbb P}\biggl( \mathbf{1}_A \int_s^t  N_{u-} d {}^o(\mathbf{1}_B( \vartheta))_u  \biggr)+ \\ && 
E_{\mathbb P} \bigl(\mathbf{1} _A \bigl \{ [N,{}^o
(\mathbf{1}_B( \vartheta))]_t - [N,{}^ o(1_B( \vartheta))]_s \bigr\} \bigr)= 
\\ && E _{\mathbb P}\biggl( \mathbf{1}_A \int_s^t \widehat{ {}^o \mathbf{1}_B(\vartheta)}_u dN_u \biggr)=
 E _{\mathbb P} \biggl( \mathbf{1}_A \int_s^t \widehat{ ( {}^o \mathbf{1}_B(\vartheta))}_u d\Lambda_u \biggr).
\end{eqnarray*} 
On the other hand, consider the counting process $\mathbf 1 _B(\vartheta)N$, which is 
adapted to $\mathbb G^{\vartheta}$,
and we know that it has a dual predictable projection with respect to $\mathbb F$:
\begin{eqnarray*} 
E_{\mathbb P}\biggl( \mathbf{1}_A  \int_s^t d(\mathbf{1}_B( \vartheta) N )_u \biggr)=
E_{\mathbb P}\biggl( \mathbf{1}_A  \int_s^t d(\mathbf{1}_B( \vartheta) N )_u^p \biggr).
\end{eqnarray*}
This means by the uniqueness of the dual predictable projection that
\begin{equation}\label{RDD}
(\mathbf{1}_B( \vartheta) N )_t^p = \int_0^t  \widehat{ ( {}^o \mathbf{1}_B(\vartheta))}_u d\Lambda_u .
\end{equation} 
We use the notation $\bar \theta $ for the measure
$$
\bar\theta (\bar X) =\bar\theta ( C\times B \times (s,t]) = 
E_{\mathbb P}\left( \mathbf 1 _C\int_s^t  \widehat{ ( {}^o \mathbf{1}_B(\vartheta))}_u d\Lambda_u\right) .
$$
extended to the $\sigma$-algebra
${\mathcal F}\otimes{\mathcal B}(\mathbb R^m) \otimes{\mathcal B}(\mathbb R^+)$.

Note that $\bar \pi $ coincides with $\bar\theta $ on the predictable $\sigma$-algebra 
$ {\mathcal P}(\bar{\mathbb{F}})$ 

Next, define a measure $\widetilde{\theta}(d\omega,d\ell,dt)$ by

\begin{equation}\label{e:two}
\widetilde{\theta}(\bar{X})=\widetilde{\theta}( C\times B \times (s,t]):=
E_{\mathbb P}\biggl(1_C \int_s^t {}^p (\mathbf{1}_B(\vartheta))_{u}d\Lambda_u \biggr)
\end{equation}
extended to the $\sigma$-algebra
${\mathcal F}\otimes{\mathcal B}(\mathbb R^m) \otimes{\mathcal B}(\mathbb R^+)$.
%with $C\in  F_s$.

Next we compare the measures $\bar \theta$ and $\widetilde \theta$
in the smaller $\sigma$-algebra ${\mathcal P}(\bar{\mathbb{F}})$,
and use Radon-Nikodym theorem to obtain a 
$\bar{\mathbb{F}}$-predictable density process.

\begin{theorem}\label{compensator}
Assume that $\bar\theta \ll \widetilde {\theta}$ on the
predictable $\sigma$-algebra ${\mathcal P}(\bar{\mathbb{F}})$ generated by the sets $A\times B \times (s,t]$ with $A\in F_s$,
and denote the 
Radon-Nikodym derivative by
\begin {eqnarray*} 
 \bar U( \omega,\ell, t)=
\frac{  d \bar \theta  }
{  d \widetilde \theta  } (\omega,\ell,t) \bigg\vert_{{\mathcal P}(\bar{\mathbb{F}})}
\end{eqnarray*}
which is $\bar{\mathbb{F}}$-predictable.
 Put $Z(\vartheta, t) = (\bar U \circ \Phi)_t$,
 and then $Z(\omega, \vartheta(\omega) , \cdot )$ is 
  $\mathbb{G}$\,-\,predictable.
Then
 we have that 
\begin {eqnarray*} 
 E_{\mathbb P}( \mathbf{1}_A \mathbf{1}_B(\vartheta) (N_t -N_s))=E_{\mathbb P}\biggl(
\mathbf{1}_A \mathbf{1}_B(\vartheta)\int_s^t  Z(\vartheta,u) d\Lambda_u\biggr)
\end{eqnarray*}
and hence  
\begin{eqnarray*} 
 N_t - \int_0^t  Z(\vartheta,u) d\Lambda_u 
\end{eqnarray*}
is a martingale in the $\mathbb{G}$-filtration.

\end{theorem} 

\begin{proof} 
The process $\bar U  $ is $\bar {\mathbb F}$- predictable
by the Radon-Nikodym theorem
, and using results of \cite{Ankirchner2005} we have that
$Z$ is $\mathbb G$-predictable. 

Now, let  $A\in F_s$, $B\in {\mathcal B}(\mathbb R^m)$, and with 
$0\le s\le t $,

\begin{eqnarray*}  
E_{\mathbb P}\left( \mathbf 1 _ A \mathbf 1 _B \left(\vartheta\right) \left(N_t-N_s\right)\right)
& = &  \bar\pi(  A\times B \times (s,t]) \\
\mbox{since } \bar\pi |\mathcal P(\bar{\mathbb F}) = \bar\theta |\mathcal P (\bar{\mathbb F} )  & =&  \int\limits_{\Omega\times\mathbb R^m \times [0,\infty) } \mathbf{1}_A(\omega) \mathbf{1}_B(\ell) \mathbf{1}_{(s,t]}(u)\bar{\theta}(d\omega,d\ell,du)\\ 
\mbox{assumption } \bar\theta \ll \widetilde \theta &=&
\int\limits_{\Omega\times\mathbb R^m \times [0,\infty) } \mathbf{1}_A(\omega) \mathbf{1}_B(\ell) \mathbf{1}_{(s,t]}(u) \bar U( \omega,\ell,u) \widetilde{\theta}(d\omega,d\ell,du) \\
\mbox{by } (\ref{e:two})  &=&
E_{\mathbb P}\left( \mathbf{1}_A \int_s^t \ ^p(\mathbf{1}_B(\vartheta) Z(\vartheta,\cdot) )_u d\Lambda_u\right)\\
&=&
E_{\mathbb P}\left( \mathbf{1}_A \int_s^t \mathbf{1}_B(\vartheta) Z(\vartheta,u) d\Lambda_u\right),
\end{eqnarray*}
where the last equality follows from the property of predictable projection \cite[Theorem V.5.16, 2)]{Chinese}.
This proves the main claim. 
\end{proof}

\begin{remark}\label{remark2}
Using the Corollary \ref{c:mrt} we get 
$$
E _{\mathbb P} (f(\vartheta) | G_t) = {}^of(\vartheta )_t = {}^of(\vartheta )_{t-} + 
\left(\widehat{ {}^o f(\vartheta) }_t - {}^of(\vartheta) _{t-}\right) \Delta N_t ;
$$ 
and this in turn gives 
\begin{eqnarray*} 
E_{\mathbb P}( f(\vartheta)\big\vert G_t) \Delta N_t =( \widehat{ {}^o f(\vartheta)})_t \Delta N_t.
\end{eqnarray*}
Therefore
we have the interpretation
\begin{eqnarray*} 
\widehat{ {}^o f(\vartheta) }_t = E_{\mathbb P}( f(\vartheta) \big\vert F_{t-} , \Delta N_t = 1 )= \frac{ E_{\mathbb P}( f(\vartheta) \Delta N_t \big\vert F_{t-} ) } 
 { E_{\mathbb P}( \Delta N_t  \big\vert F_{t-} ) }.
\end{eqnarray*} 
\end{remark}

\begin{remark} \label{remark3}
We give an interpretation of the condition $\bar \theta \ll \widetilde \theta $.

First, consider the formal disintegration of measure

 \begin{eqnarray*} 
{\bar\theta } (d\omega,d\ell,dt) = \mathbb{P}(d\omega)\bar \nu(d\ell,dt;\omega)
\end{eqnarray*} 

Here  we can interpret 

\begin{eqnarray*}  
\bar \nu(d\ell ,dt) &=& \mathbb{P}(\vartheta\in d\ell,N(dt)=1 \big\vert F_{t-}) \\
& = & \mathbb{P}(\vartheta\in d\ell\big\vert F_{t-})P( N(dt) = 1 \big\vert F_{t-},\vartheta \in d \ell)\\
& =: &\mathbb{P}(\vartheta\in d\ell \big\vert F_{t-})\Lambda^{\ell}(dt). 
\end{eqnarray*}

On the other hand, 
 we also have the disintegration
\begin{eqnarray*} 
\widetilde {\theta } (d\omega,d\ell,dt)  =\widetilde \nu(d\ell,dt;\omega)\mathbb{P}(d\omega) ,
\end{eqnarray*} 
and from (\ref{e:two}) we have 
 \begin{eqnarray*} 
\widetilde\nu(d\ell,dt) =  \mathbb{P}(\vartheta\in d\ell \big\vert F_{t-})\Lambda(dt) .
\end{eqnarray*} 

Now, if  $\bar\theta \ll \widetilde {\theta}$ then 

\begin{eqnarray*} 
\bar U ( \omega,\ell, t)&=&\frac{  d \bar\theta  }
{  d\widetilde \theta } (\omega,\ell,t) 
 \bigg\vert_{{\mathcal P}(\bar{\mathbb{F}})}\\
& =&  \frac{d\bar\nu}
{d\widetilde\nu}(\ell,t) = \frac{ d \mathbb P ( \vartheta \in \cdot \;| F_{t-}, \Delta N_t=1 ) }
{ d \mathbb P( \vartheta \in \cdot \; | F_{t-} ) } ( \ell,\omega) =
 \frac{d\Lambda_t^{\ell} }{d\Lambda_t}(t,\omega).
\end{eqnarray*}

Moreover, we have the connection 
$$
\Lambda ^{\ell}_t = \int _0^t Z(\ell,s)d \Lambda _s .
$$

\end{remark}

\begin{remark}
When the absolute continuity condition fails,
 the Lebesgue decomposition  on  $ {\mathcal P}(\bar{\mathbb{F}})$
\begin{eqnarray*}
  \bar\theta( d\omega,d\ell,dt)= {\bf 1}(  U ( \omega,\ell, t)  < \infty )U ( \omega,\ell, t)  \widetilde\theta( d\omega,d\ell,dt)
+ {\bf 1}(  U ( \omega,\ell, t)  = \infty )\bar\theta( d\omega,d\ell,dt) 
\end{eqnarray*} corresponds  to the Lebesgue decomposition of the $\mathbb{G}$-compensator
\begin{eqnarray*}
 \Lambda ^{\vartheta }(dt) =  {\bf 1}( Z(\vartheta ,t)<\infty) Z(\vartheta ,t) \Lambda(dt) + {\bf 1}( Z(\vartheta ,t)=\infty)  \Lambda ^{\vartheta }(dt)
\end{eqnarray*}
\end{remark}

About the singular part of $\Lambda^{\vartheta }$, at this level of generality we cannot say much more than this:

\begin{proposition} \label{singular:part}
In the   $\mathbb{G}$-filtration the jumps of $N$ are decomposed 
into two classes, $\mathbb{G}$-accessible and  $\mathbb{G}$-totally inaccessible
\cite[Chapter IV]{Chinese}.
%(see section \ref{classification:stoppingtimes} below).
The next conditions are equivalent: 
\begin{itemize}
 \item   $\mathbb P$-almost surely 
\begin{eqnarray}
 \label{singular:compensator}
\int_0^t {\bf 1}\bigl( Z(\vartheta ,s)=\infty\bigr) \Lambda^{\vartheta } (ds)
= \sum_{s\le t} {\bf 1}\bigl(   Z(\vartheta,s)=\infty \bigr) \Delta\Lambda^{\vartheta }_s
\end{eqnarray}
i.e. the singular part of the $\mathbb{G}$-compensator is purely discontinuous
\item
The  $\mathbb{G}$-compensator of the $\mathbb{G}$-totally inaccessible part of $N$
 is  absolutely continuous w.r.t. $\Lambda$,
and (\ref{singular:compensator})
is the  $\mathbb{G}$-compensator of the
  $\mathbb{G}$-accessible jumps of $N$.
\end{itemize}
\end{proposition}
\begin{remark} We will see that in our initially enlarged Markov Chain market model 
we are in the situation described in Proposition
 \ref{singular:part}.
\end{remark}

\section{Scenarios and support of the predictive distribution}

\subsection{Shrinkage}\label{ss:shrinkage}
We start with an useful Lemma, which helps to compute compensators. 
We assume now that the filtration $\mathbb F$ is the filtration of the 
Markov process $Y$, and the random variable in the initial enlargement is the 
logarithm of final value of the stock: $\vartheta = \log (S_T)$: 
$F_t = \sigma \{ Y_s : s\le t \}$ and $G_t = \cap_{u> t} F_u \vee \sigma (\vartheta ) $.  
In addition to the random variable $\vartheta $ we enlarge the filtration $\mathbb F$ 
with the realized scenario $\zeta = H_T $,
where $H_T = (N_T ; Y_0, Y_{\tau _1 }, \dots, Y_{\tau _{N_T}}) $ 
(see \ref{ss:scenario} for more details). Note that the random variable $\chi$ can take only 
countably many values.

\begin{lemma} \label{countable-shrinkage}  
Assume that $(\vartheta,\zeta )\in F^Y_T$,  
and the random variable $\zeta $ takes in a countable set, say $
\zeta(\omega)\in {\mathbb{Z}}$
without loss of generality.

Let $\mathbb{G} ^\zeta =(G_t^{\zeta })_{t\in[0,T]}$ 
with  $G^{\zeta }_t= \cap _{u > t} G_u \vee\sigma(\zeta )$ be a bigger filtration than 
$\mathbb{G}$. Then we have the \emph{filtration shrinkage } formula
\begin{eqnarray*} 
 \Lambda^{\vartheta}_t =(\Lambda^{\vartheta,\zeta  })_t^{\mathbb{G}, p} = 
\sum_{z} \int_0^t \mathbb{P}( \zeta = z \big\vert F_{s-}, \vartheta ) \Lambda^{\vartheta,\zeta  =z} (ds) 
\end{eqnarray*} 
where the  $\mathbb{G}$- predictable processes $\{ \Lambda^{\vartheta,\zeta=z}_t(\omega) :\, z\in \mathbb{Z}\; 
\mathrm{and}\; \mathbb{P}(\zeta =z)>0\}$ 
gives the disintegration of the $\mathbb {G}^{\zeta}$ compensator.
\begin{eqnarray*} 
\Lambda^{\vartheta,\zeta }_t  =  
\sum_{z: \mathbb{P}(\zeta =z)>0 } \int_0^s  {\bf 1}_{\{ \zeta = z \}}  \Lambda^{\vartheta,\zeta =
z} (ds)
\end{eqnarray*} 
\end{lemma} 

\begin{proof}
For more general results of this type, see \cite{Protter2005}. We prove the result in this simple case. 
Let $s\le t $ and $A\in G_s$. We have
\begin{eqnarray*} 
 E_{\mathbb P}\bigl( 1_A  (N_t -N_s)\bigr) &= & 
 E_{\mathbb P}\biggl( \sum_{z }
  {\bf 1}_{ A \cap\{  \zeta = z \} }( N_t -N_s )\biggr)\\
& = &
  E_{\mathbb P}\biggl( \sum_{z}  {\bf 1}_{ A \cap\{  \zeta= z \} }  (\Lambda^{\vartheta, \zeta}_t -\Lambda_s^{\vartheta,\zeta }) \biggr)
 \end{eqnarray*}
 where the sum is taken over the values $z $ with  $\mathbb{P}( \zeta =z)>0$. But on the set $\{ \zeta = z\} $ we have the identity 
${\bf 1}\{ \zeta =z\} \Lambda ^{\vartheta,\zeta}_u = {\bf 1}\{ \zeta =z\} \Lambda ^{\vartheta,z}_u $.  We obtain
% observing that the indicator  ${\bf 1}_{ A \cap\{  H= h \} }$ is  ${\mathcal G^{L,H} }$-measurable,
\begin{eqnarray*}  
 E_{\mathbb P}\biggl( \sum_{z}  {\bf 1}_{ A \cap\{  \zeta = z \} }  (\Lambda^{\vartheta,\zeta}_t -\Lambda_s^{\vartheta,\zeta }) \biggr) &= &    
 E_{\mathbb P}\biggl(   \sum_{z} {\bf 1}_{ A \cap\{  \zeta= z \} }  (\Lambda^{\vartheta,\zeta =z}_t -\Lambda_s^{\vartheta,\zeta =z}) \biggr) \\
   E_{\mathbb P}\biggl( \sum_{z} \int_s^t  { }^{\mathbb{G}, p }({\bf 1}_{ A \cap\{  \zeta = z \} })_u  d\Lambda^{\vartheta,\zeta =z}_u \biggr)  &=&
E_{\mathbb P}\biggl(  \sum_{z}  1_A \int_s^t  { }^{\mathbb{G}, p } ({\bf 1}_{ \{  \zeta = z \} })_u  d\Lambda^{\vartheta,\zeta =z}_u \biggr)\\ &=&
E_{\mathbb P}\biggl(  \sum_{z}  1_A \int_s^t  \mathbb{P}(  \zeta= z \big\vert G_{u-} ) d\Lambda^{\vartheta,\zeta =z}_u \biggr) 
 \end{eqnarray*}
since   $\Lambda ^{\vartheta, \zeta= z}$ is $\mathbb{G} $- predictable,
 $A \in G_s $, and by the definition of predictable projection.
\end{proof}

\begin{remark} Lemma \ref{countable-shrinkage} gives a way to compute the
compensator $\Lambda^\vartheta$ by  using an additional countable enlargement.
We have also
\begin{eqnarray*} && 
 \Lambda^\vartheta\ll \Lambda\quad\quad \Longleftrightarrow \quad\quad  
\Lambda^{\vartheta,\zeta} \ll \Lambda
\\ &&\Longleftrightarrow \quad\quad\Lambda^{\vartheta,z} \ll  \Lambda 
\quad\quad\forall\, z\in \mathbb Z\;\; \mathrm{and}\;\; \mathbb{P}(\zeta =z)>0
\end{eqnarray*} 
\end{remark}

\subsection{More on scenarios}
\subsubsection{Random scenarios}
We have already introduced the notion of \emph{scenario} in the Section \ref{ss:scenario}. Now 
we will assume that in addition to the final value the insider has  at his disposal the information, 
which states the Markov process $Y$ visited before the time $T$. 

Let $\Xi _{e_0}$ be the set of all possible scenarios starting from $e_0$ and let $\Xi$ be the set of all 
possible scenarios: 
$$
\Xi_{e_0}= \{ h= (n; e^{0:n}) : n \in \mathbb N , \lambda ^{e_ie_{i+1}} > 0, i= 0, \dots , n-1 \}
$$   and  $ \Xi=\cup_{e_0\in \mathcal Y  } \Xi_{e_0} $. Note that $\Xi $ is numerable. 

Recall that $\tau _i $ is the $i^{\mbox{th}}$ jump time of the economy and then $H_t$ is the random scenario
$$
H _t(\omega ) = (N_t(\omega ) : Y_0, Y_{\tau _1}, \dots, Y_{\tau _{N_t}} ).
$$
and  $ H_{T}$ is the random scenario $H_T  = (N_T : Y _0, Y_{\tau _1}, \dots , Y_{\tau _{N_T}} )$.

\subsubsection{Operations with scenarios}
%Next we discuss various properties of scenarios and some operations with them.
To analyze the scenarios dynamically we need the following operations with them.
Let $h= (n;e^{0:n}) $ and $\widetilde h= (m;{\widetilde e}^{0:m})$ be two scenarios. Put 
$$
h_{(k)} = (n\wedge k;  e^{0:n\wedge k}) , h^{(k)} = ((n-k)^+ ; e^{n\wedge k:n}) 
$$
and 
$
h\vee \widetilde h = (n+m; e^{0:n},{\widetilde e}^{0:m}) $; here we assume that $e^n = \widetilde e^0$.

With these notations $h^{(0)}=h= h_{(k)} \vee h^{(k) } =h_{(n)}\vee h^{(n)}$.

 Let $h= (n; e^{0:n}) $ be a fixed scenario, and put 
\begin{eqnarray*}
\Pi_{e_0,t}(h):= \mathbb{P}({H}_t= h \big\vert Y_0=e_0) 
\end{eqnarray*}
Note that for every $h \in \Xi_{e_0} $ we have that $\Pi _{e_0,T}(h) > 0 $.

We have $ \Pi_{e_0,T}((0;e_0))= \exp(- \lambda^{e_0} T )$ and with $h= (n; e^{0:n}) $, when 
$n\ge 1 $, we have the recursion:
$$
\Pi_{e_0,T}( h ) = \int_0^T  \lambda^{e_0,e_1}
  \exp(- \lambda^{e_0} t ) \Pi_{e_1,{T-t}}( h^{(1)}) dt .
$$

To summarize what we have achieved by now:

$\Pi_{e_0,T}( h ) >0 $ if and only if $h\in \Xi_{e_0}$  
and we have the implications
\begin{eqnarray*}
 \Pi_{e_0,T}( h )>0\quad  \Longrightarrow\quad  \Pi_{e_0,t}(h)>0 \quad  \mbox{  for all }   t>0 ;
\end{eqnarray*}
this means that if a fixed  scenario $h$  has positive probability on the interval $[0,T]$, 
it has a positive  probability on every sub-interval $[0,t]$, too. Finally, using the 
identity  $h= (h_{(k)}\vee h^{(k)})$, we have the following implications for all $t\in (0,T)$: 
\begin{eqnarray*}
 \Pi_{e_0,T}(h)> 0 \quad \Longrightarrow \Pi_{e_0,t}(h_{(k)}) > 0, \mbox{ and }  \\
 \Pi_{e_0,T}( h )>0\quad  \Longrightarrow\quad  \Pi_{e_k,T-t}(h^{(k)})>0 .
\end{eqnarray*}

\subsection{Joint distribution of $L_t$ and $H_t$}

Recall that $L_t = \log ( S_ t ) $, where $ S$ is the discounted stock vector, and 
$  H_t $ is the random scenario $H_t = (N_t; Y_0, \dots , Y_{N_t})$.  We denote their joint 
distribution by $Q$.

Put $Q_{e_0,t}( d\ell, (0;e_0) ) := \mathbb{P}(L_t \in d\ell, {H}_t=(0;e_0))$. 
Note first that 
\begin{eqnarray*}
Q_{e_0,t}( d\ell, (0;e_0) )  = \exp( -\lambda^{e_0}  t) \delta_{ \mu^{e_0}t } (d\ell) .\end{eqnarray*}
After this we can proceed recursively
\begin{eqnarray}\label{recursive_measure}
 Q_{e_0,t} (d\ell,h)  &= &  \mathbb{P}(L \in d\ell, {H}_t = h)\nonumber \\ &  = &
\int_0^t  \lambda^{e_0,e_1}\exp( -\lambda^{e_0} u) Q_{e_1,t-u}(d\ell-\mu^{e_0}u-  \beta^{e_0,e_1},h^{(1)})\,du.
\end{eqnarray}
From the joint distribution $Q_{e_0, T } (d\ell , h) $ we obtain the marginal distribution 
\begin{eqnarray*}
Q_{e_0, T} (d \ell ) = \sum _{h\in \Xi _{e_0} } Q_{e_0,T} (d\ell, h ) . 
\end{eqnarray*}
\subsection{Support of the conditional measure $Q_{e_0,t} ( \cdot | h ) $}

When $h \in \Xi_{e_0}$, then the conditional probability 
$Q_{e_0,t}(d\ell | h)= \frac{Q_{e_0, t}(d \ell,h)}{\Pi_{e_0,t}(h)}$ is well defined, since $\Pi_{e_0,t}(h) > 0 $ for $0<t\le T $.

Fix $h = (n; e^{0:n})$ and put $\beta ^{\overline{0:n}} = \beta ^{e_0e_1} + \cdots + \beta ^{e_{n-1}e_n} \in \mathbb R ^m$. 
The support of the conditional measure 
 $Q_{e_0,T}(d\ell\big\vert h)$ is obviously the convex hull of the set
\begin{equation}\label{eq:pre-supp}
\bigl\{ \; L_0+ \beta ^{\overline{0:n}} +\mu^{e_i}T \;: \; i=0,\dots, n \; \bigr\}.
\end{equation}
We denote this convex hull by $\mathcal A _T(h) $. Fix $0<s<T$, and consider the convex hull 
$\mathcal A_{T-s}(h^{(N_s)})$
of the random set 
\begin{equation}\label{eq:pre-supp-s}
 \bigl\{ \; L_s + \beta ^{\overline{N_s:n}} +\mu^{e_i} (T-s) \;: \; i=N_s,\dots, n \; \bigr\}.
\end{equation}

Then, either $Q_{e_0,T}(d\ell\big\vert h)$ is a point mass, which happens if and 
only if  $ \mu^{e_i} = \mu^{e_0}$ 
for all $i=1,\dots,n$, or $Q_{e_0,T}(d\ell\big\vert h)$ is equivalent to  Lebesgue measure on its support.

Moreover, we have 
for every $\omega$ in the canonical space, $0<s<t<T$,
  $h\in \Xi$:
\begin{eqnarray*} 
 \mathbb R^m\supset \mathcal A_T(h) \supseteq \mathcal A_{T-s}(h^{(N_s)}) \supseteq \mathcal A_{T-t}(h^{(N_t)}) ,
\end{eqnarray*} 
and by summing over the  scenarios $h\in \Xi$ we get
\begin{eqnarray*} 
 \mathbb R^m\supset\mbox{\rm supp}\;Q_{Y_s,T-s}(\;\cdot -L_s ) \supseteq 
\mbox{\rm supp}\; Q_{Y_t,T-t}(\;\cdot -L_t ) 
\end{eqnarray*} 
Since these predictive distributions are equivalent to Lebesgue measure on their support, 
the relation 
\begin{eqnarray*} 
 \mbox{\rm supp}\; Q_{Y_s,T-s}(\;\cdot -L_s | h^{(N_s)} ) 
\supseteq \mbox{\rm supp}\; Q_{Y_t,T-t}(\;\cdot -L_t \vert h^{(N_t)} )
\end{eqnarray*} 
does not imply that
\begin{eqnarray*} 
Q_{Y_s,T-s}(\;\cdot -L_s | h^{(N_s)} ) \gg Q_{Y_t,T-t}(\;\cdot -L_t \vert h^{(N_t)} ) .
\end{eqnarray*} 
This implication is true only in the case of that 
the supports of these predictive distributions have the same dimension.

More formally, put 
\begin{eqnarray*} &&
 D_T(h)   =   \mbox{dim} \mathcal A_{T}(h) \\
&&  \;=
 \max \{ p : \exists x_0,x_1,\dots, x_p \in \mathcal A_{T}(h) \mbox{ with $(x_k-x_0)$ linearly independent} \} ,
\end{eqnarray*}
where $D_T(h)= 0$, if the set ${\mathcal A}_{T}(h)$ consists of one point. Note that the value of $D_T(h) $ does not depend on 
$T$, and we write simply $D(h) $. 

Fix now $h= (n: e^{0:n}) $ and assume that $n\ge 1$. Recall that $h^{(k)}$ is the remaining scenario after 
$k$ changes in the economy. Obviously we have 
\begin{equation}\label{eq:d-k}
 D(h) \ge D(h^{(1)}) \ge D(h^{(2)}) \ge \cdots \ge D(h^{(n)}) = 0 .
\end{equation}
 
Clearly, if $D(h)=0$, then $D(h)= D(h^{(1)} ) = \cdots = D(h^{(n)}) = 0 $.
\begin{example}\label{ex:k-h-a}
 Returning to the example \ref{ex:one}   of Kohatsu-Higa we have for $0<s<T$ that
$$
\mathcal A _T(h) = \{ L_0+\beta ^{\overline{0:n}} +  \mu^{0:n} T \} = \mathcal A_{T-s}(h^{(N_s)}).
$$
hence we have $D(h) = D(h^{(1)} ) = \cdots = D(h^{(n)})=0 $. 
\end{example}

\begin{example}\label{ex:two-d}
Take $h= (n;e^{0:n}) $ with $n\ge 1 $. Then $n= n^+ + n^- $, where $n^- = \lfloor \frac{n}{2}\rfloor$, and 
$$
D(h) = \cdots = D(h^{(n-1)}) = 1 > D(h^{(n)})=0. 
$$ 
\end{example}

\subsection{Scenarios and final value}

In order to analyze, how the properties of the jump times may change with the additional information, we need more definitions.

For given $T > 0,\ell\in \mathbb R^m $, consider those scenarios $h\in \Xi_{e_0}$, $h=(n;e^{0:n})$, 
$n\in \mathbb N$, such that $\ell \in  {\mathcal A}_T(h)$, that is 
 for some $\Delta t_j  >  0 , j=0,1,\dots,n$
\begin{eqnarray}\label{eq:pred}
{\begin{matrix}
 \sum _{j=0}^n  \mu ^{e_j} 
\Delta t_j =\ell-L_0-\beta^{\overline{0:n}}, \quad\quad \mbox{ and } \\ 
\Delta t_0 +\Delta t_1 + \dots+ \Delta t_n= T \; .
\end{matrix}} 
\end{eqnarray}
Note that for given $(T,\ell,h)$, the solution vector $(\Delta t_j: 0\le j\le n)$ possibly does not
exist, and when it exists, it is not always unique.

Consider the projection
 \begin{eqnarray} &&
C_0(h):= \bigl\{ \Delta t_0>0 : \exists(\Delta t_0,\dots, \Delta t_n) \in \mathbb R_+^{n+1} 
\;\mbox{ solving }(\mbox{\ref{eq:pred}}) \bigr\} 
\\ &&
\label{GH_predictable_times}
\underline{\mathcal T}(T,\ell,h):=\inf C_0(h), \quad \overline {\mathcal T}(T,\ell,h):=\sup C_0(h)
\end{eqnarray}
When $C_0(h)\ne \emptyset$, since the solutions of (\ref{eq:pred}) form a convex set,
\begin{itemize}
\item
either $\underline {\mathcal T}(T,\ell,h) < \overline {\mathcal T}(T,\ell,h)$ and 
$C_0(h)=\bigl( \underline{\mathcal T}(T,\ell,h) , \overline {\mathcal T}(T,\ell,h) \bigr)$,
\item
or $\underline {\mathcal T}(T,\ell,h) = \overline {\mathcal T}(T,\ell,h) $ and
 $\Delta t_0$  is \emph{determined}
 by $(T,\ell)$ and the scenario 
$h=(n;e^{0:n})$.
\end{itemize}

\begin{proposition}
% Fix $T$ and $h$.
 We have the following characterization:
\begin{equation}\label{eq:char}\Delta t_0  \mbox{ is determined by }
 (T, \ell, h )  \Leftrightarrow \ell \in  {\mathcal A}_T(h)  \mbox{ and }
D(h^{(1)} ) + 1 = D(h) 
\end{equation}
\end{proposition}
\begin{proof} 
The constrained linear problem ${\ref{eq:pred}}$ is rewritten as
\begin{eqnarray}\label{eq:pred:reformulated} &&
  \sum_{i=0}^{n-1} (\mu_i -\mu_n) \Delta t_i = (\ell - L_0-\beta ^{\overline{0:n}} -\mu_n T)
 , \; \Delta t_i >0 ,\sum_{i=0}^{n-1} \Delta t_i < T 
\end{eqnarray}
Consider the constrained linear systems	
\begin{eqnarray*}   (A):  && 
  \sum_{i=0}^{n-1} A_{ji} t_i = y_j, \; 1\le j\le m ,\; t\in C\subseteq \mathbb R^n
\\  (A'): &&
  \sum_{i=1}^{n-1} A_{ji} t_i' =y'_j, \; 1\le j \le m, \; t' \in C'\subseteq \mathbb R^{n-1}
\end{eqnarray*}
where $C$ and $C'$ are open simplexes (cf. ${\ref{eq:pred}}$) .

$(A)$ corresponds to ({\ref{eq:pred:reformulated}}) and
 $(A')$ corresponds to the situation after 
the first transition $e_0 \to e_1$.
Denote 
\begin{eqnarray*} 
A':=\bigl(A_{ji} \bigr)_{ 1\le i \le (n-1), 1\le j \le m} 
\end{eqnarray*}
The images  $AC$ and $A'C'$ are open in $\mathbb R^m$,
and their dimension coincides with the ranks
$\dim( {\rm Im}( A) )=D(h)$ and $\dim( {\rm Im}(A'))=D(h^{(1)})$ respectively.

We have the  linear isomorphisms
 \begin{eqnarray*} 
  { \rm Im} ( A ) \simeq  \bigl( \mathbb R^n \big / {\rm Ker}( A ) \bigr), \quad  {\rm Im} (A' ) \simeq \bigl(
\mathbb R^{n-1} \big / {\rm Ker}( A' ) \bigr)
 \end{eqnarray*}
where ${\rm Ker}(A)$ denotes the null space  and we take the algebraic quotient. This implies
\begin{eqnarray*} 
 \dim ( {\rm Im} (A ) ) = n - \dim({\rm Ker}(A) ), \quad \dim ( {\rm  Im} (A' ) ) = n-1 - \dim({\rm Ker}(A') )
\end{eqnarray*}
Either
\begin{enumerate}
 \item 
the column vector $A_{\bullet 0}$ is linearly independent from
the columns $( A_{\bullet 1}, \dots ,A_{\bullet n-1})$
\begin{eqnarray*} 
 \Longleftrightarrow \dim({\rm Im}(A))=\dim ( { \rm Im} (A' ) )+1 
 \Longleftrightarrow  \dim({\rm  Ker}(A))=\dim( {\rm Ker}(A')) ,
\end{eqnarray*}
\item
 or $\dim({\rm Im}(A))=\dim ( {\rm Im} (A' ) )$,
\begin{eqnarray*}
\Longleftrightarrow\dim( {\rm Ker}(A))=\dim( \mbox{\rm Ker}(A'))+1 .
\end{eqnarray*}
\end{enumerate}
In  case (1),
the dimension of the null space does not change after adding the  column
 $A_{\bullet 0}$  to the matrix $A'$.
If $( t_0, t_1, \dots,  t_{n-1})$ is a solution of the homogeneous system associated to $ (A)$,
\begin{eqnarray*}  
 (A^*): && \sum_{i=0}^{n-1} A_{ji} t_i = 0 , \; 1\le j \le m
\end{eqnarray*}
necessarily $t_0 =0$ and $ \sum_{i=1}^{n-1} A_{ji}  t_i =0, \; \forall j$.
This means that all solutions of $(A)$ begin with the same coordinate $ t_0$.

In case (2) the homogeneous system $(A^*)$ admits solutions with $ t_0\ne 0$ 
and $ t_0$ is not uniquely determined by $(A)$.
\end{proof}

\begin{example}
Consider  the example \ref{ex:k-h-a}, fix $h$ and take $\ell \in \mathcal A_T(h) $. It is easy to see, that $(T , \ell , h) $ never 
determines $\Delta t_0 $. On the other hand, in the example \ref{ex:two-d}, take $h= (1; e^{1}, e^{2})$, and $\ell \in \mathcal A_T(h) $.
Then $(T, \ell , h) $ determines $\Delta t_0 $ and we have 
$$
\Delta t_0 = \frac{\ell  -L_0 - \beta ^+ - \mu ^+ T }{\mu ^- -\mu ^+ }.
$$
More generally, in this example, for any $h$ with $n$ changes in the economy, and $\ell \in  \mathcal A _T( h) $, the last jump time 
$\sum _{k=0} ^{n-1} \Delta t_k $ is known, if we know the value of $\sum _{k=0}^{n-2}\Delta t_k $. 
\end{example}

\section{Computation of the insiders compensator}

Our program has two parts:
i) Obtain information about the compensator of $N$ with respect to the filtration  $\mathbb G $.
ii) Check the (NA) criteria in the enlarged filtration $\mathbb G$.

%We must work without any hypothesis (H), i.e. without so called Jacod's hypothesis for reasons which will be clear later.

\medskip

The idea is to enlarge the filtration $\mathbb G$ with the information of the random scenario $H_T$,
and then use filtration shrinkage to obtain the compensator with respect to  $\mathbb G$.  

%\section{Computation of the insiders compensator}

\subsection{Classification of the jump times in an extended filtration }
We work with the filtration $\mathbb G ^{\mathbb H} $, where 
$G_t^{\mathbb H} = \bigcap _{u> t}G_u \vee \sigma (H_T) $.

The following proposition a summary of the results in the previous section. 

\begin{proposition} Consider the $k$-th jump time $\tau_k$. Fix an history $h$.
 on the set  $\{ \omega: H_T(\omega)=h\}$,
\begin{itemize}

\item[(a)]
either  $D(h^{(k-1)}_T)=D(h^{(k) }_T)$, so that  $\forall s \in (\tau_{k-1}, \tau_k]$,
\begin{eqnarray*} 
 Q_{e_k,T-s}(d\ell -(\mu^{e_{k-1}}-r^{e_{k-1}} )(T- s)- \beta^{e_{k-1},e_k},
h^{(k)})  \ll  Q_{e_{k-1},T-s} (d\ell, h^{(k-1)} ) 
\end{eqnarray*} 
with
\begin{eqnarray*}  &&
{\bf 1}( \tau_{k-1}< s \le \tau_k ) \Lambda^{ L_T,h}(ds)   = \\ &&
{\bf 1}(  \tau_{k-1}< s \le \tau_k )
 q_{(e_k,T-s,h)}(L_T - L_{s-}-\beta^{e_{k-1}e_k})
\lambda^{e_{k-1},e_k}  ds \; ,
\end{eqnarray*} 
where 
\begin{eqnarray*}
q_{(e_k,T-s,h)} (\ell ) := \frac{  d Q_{e_k,T-s}(\;\cdot -(\mu^{e_{k-1}}-r^{e_{k-1}} )(T- s)- \beta^{e_{k-1},e_k},
h^{(k)})}{ d Q_{\;e_{k-1},T-s} ( \cdot ,h^{(k-1)}  ) }(\ell ) 
\end{eqnarray*}
is supported by the random interval $( \underline \tau_k(h) ,\overline \tau_k(h) ]$.
where
 by using (\ref{GH_predictable_times}) we define the
$\mathbb G $-predictable times 
\begin{eqnarray*}
\underline\tau_k(h) & := &
  \tau_{k-1}+ \underline
{\mathcal T } ( T-\tau _{k-1}, L_T - L_{\tau _{k-1}},  h^{(k-1)}_T ),
\\
\overline\tau_k(h) & := &
  \tau_{k-1}+ \overline
{\mathcal T } ( T-\tau _{k-1}, L_T - L_{\tau _{k-1}},  h^{(k-1)}_T ),
\end{eqnarray*}
which are also $\bigl( \mathbb G_{\tau_{(k-1)}}\bigr)$-measurable and
satisfy 
 \begin{eqnarray*}
\tau_{k-1}\le \underline \tau_k(h) \le \tau_k \le \overline\tau_k(h) 
\quad \mbox{ on }\{ \omega: H_T(\omega)=h\}
\end{eqnarray*}

 \item[(b)] 
or 
 $D(h^{(k-1)}_T)=D(h^{(k) }_T)+1$, so that  $\tau_k=\underline \tau_k(h)=
\overline \tau_k(h) $ .
\end{itemize}
\medskip

When we sum over all scenarios  $h$ we obtain  
\begin{itemize}
 \item[(A)] : when $D( H_T^{(k-1)} )=D( H_T^{(k) })$, $\tau_k$ has 
$\mathbb G^{\mathbb H}$-compensator
absolutely continuous w.r.t. $\Lambda$ and
  \begin{eqnarray*}
\tau_{k-1}\le \underline \tau_k(H_T ) \le \tau_k
 \le \overline\tau_k(H_T ) 
\end{eqnarray*}

\item[(B)] : otherwise  $\tau_k=\underline   \tau_k(H_T) = \overline \tau_k(H_T)$ 
\end{itemize}
where $\underline \tau_k(H_T ),\overline \tau_k(H_T )$ 
are   $\mathbb G^{\mathbb H}$-predictable times.

%In this case
%${\bf 1}( \tau_k \le t)$ is the 
%$\mathbb G ^{\mathbb H}$-compensator of itself.

\end{proposition}

\begin{corollary} 
 $\Lambda ^{L_T, \mathbb {H}}$ is absolutely continuous
w.r.t. $\Lambda$  if and only if $D( H_T) =0 $.
\end{corollary} 

\vspace{15pt}
Now, we apply the countable filtration shrinkage argument to the
$\mathbb G ^{\mathbb H} $- compensator of $N$ to
obtain the $\mathbb G $-compensator.

\begin{proposition} \label{G^L-insider-compensator}  The  $\mathbb G $-compensator of $\tau_k$ is given by
\begin{eqnarray*}   && 
 \int_0^t {\bf 1}( \tau_{k-1} < s \le \tau_k) \Lambda^L(ds)=
\\ &&
 \sum_{h\in {\mathcal D}_k}\int_0^t \mathbb{P}
( H_T=h \big\vert L_T, F_{s-} ) {\bf 1 }( \tau_{k-1} < s \le \tau_k ) d_s \bigl( {\bf 1 }( 
 \overline \tau_k(h) \le  s) \bigr)
 + \\ &&
 + \sum_{h\in \Xi \setminus  {\mathcal D}_k }  \int_0^t 
\mathbb{P}( H_T=h \big\vert L_T, F_{s-} ) 
{\bf 1 }( \tau_{k-1} < s \le \tau_k )
\Lambda^{L,h}(ds)
\end{eqnarray*} 
where 
\begin{eqnarray*}   
 {\mathcal D}_k & =& \bigl\{ h \in \Xi : D(h^{(k-1)} ) = D( h^{(k) })+1 \bigr\},
\end{eqnarray*} 
is the set of scenarios for which $\tau_{k}$ is determined by $L_T$ and $H_T$ at time $\tau_{k-1}$,
and $\Lambda^{L,h}(ds) \ll \Lambda(ds)$ for $s \in (\tau_{k-1} , \tau_k]$
 and $h\in  \Xi \setminus  {\mathcal D}_k$.

This gives the decomposition of $\tau_k$ into
 ${\mathbb G}$-accessible and ${\mathbb G}$-totally inaccessible parts.

\end{proposition} 

%Proposition \ref{G^L-insider-compensator} gives a decomposition of  $\tau_k$ into $\mathbb{F}^\vartheta$- accessible and  totally inaccessible parts.

Note also  that the predictable times $\{ \overline \tau_k(h) : h \in {\mathcal D}_k\}$ are not
necessarily distinct. 
Let $ {\mathcal D}_k^* \subseteq {\mathcal D}_k$ a choice of distinct representatives
w.r.t. the equivalence relation 
\begin{eqnarray*} 
 h \stackrel{k}{ \sim }h' \Longleftrightarrow \overline  \tau_k(h) =  \overline \tau_k(h'), \; h,h'\in   {\mathcal D}_k
\end{eqnarray*}

By re-summation, the compensator
  $\mathbb{G}$-accessible  part of the stopping time $\tau_k$  is rewritten as 
\begin{eqnarray*}
\sum_{h\in {\mathcal D}_k^*}\int_0^t 
\biggl\{ 
\sum_{h'\in {\mathcal D}_k:  \overline\tau_k(h') =  \overline \tau_k(h) } 
\mathbb{P}( H_T=h' \big\vert L_T, F_{s-} ) \biggr\} d_s \bigl( {\bf 1 }( \overline  \tau_k(h) \le  s) \bigr)
 \end{eqnarray*}
where the ${\mathbb G}$-predictable jump times $\{ \overline \tau_k(h) : h \in D_k^* \}$ are distinct.

\begin{example}
 Concerning the example \ref{ex:two} there are two possibilities. The final value $L_T=\ell$ does not uniquely 
determine the scenario $H_T= h$. In this case the compensator is totally inaccessible in the filtration 
$\mathbb G$. But with special parameter values $\mu^\pm $ and $\beta ^\pm $ $H_T$ is uniquely determined by 
$L_T= \ell $, and then for the insider the last jump is predictable. 
\end{example}

\section{ Insider's  Free lunch with vanishing risk } 
From the general theory it follows that the property
{\it No free lunch with vanishing risk }(NFLVR) in the insider filtration $\mathbb{G}$
is equivalent to the existence of a measure $\mathbb{Q}^L\sim \mathbb{P}$  
under which the discounted stock process $(\widetilde S_t)_{t\geq 0}$ is 
 a $\mathbb{F}^\vartheta$-martingale. This leads to  conditions concerning 
the accessible and totally inaccessible parts of the jumps  of $(L_t)_{t\geq 0}$. 
We also see  that, for arbitrage considerations, we do not need to fully compute the compensators in
the insider filtration:  it is enough to compute the random sets ${\mathcal D}_{k}$ at each 
jump time  $\tau_{k-1}$.
\\

For ${\mathcal A}\subseteq \mathcal{Y}^e $ we consider the system of equations
\begin{equation}
\label{matrix system2}
\mathbf{\Gamma}^{e,\mathcal A} \widetilde{\mathbf{\lambda}}^{e,\mathcal A}=-\mathbf{\mu}^e ,
\end{equation}
and the homogeneous system
\begin{equation}
\label{matrix system3}
\mathbf{\Gamma}^{e,\mathcal A} \widetilde{\mathbf{\lambda}}^{e,\mathcal A} = \bf{ 0 }
\end{equation}

where  $\mu^e=\left( \mu^{ie}\right)_{i=1;\ldots,m}$,
\begin{equation}
\mathbf{\Gamma}^{e,\mathcal A} :=(
{\mathbf{\gamma} }^{ief} : i=1,\dots,m, f\in {\mathcal A})
\end{equation}
\begin{equation}
\widetilde{\mathbf{\lambda}}^{e,{\mathcal A}} =
( \lambda^{e,f} :  f\in {\mathcal A} )
\end{equation}
with the constraints 
\begin{equation}
\widetilde \lambda^{e,f} > 0   \mbox{ strictly  for }    f\in {\mathcal A}
\end{equation}
This means that respectively $(-\mathbf{\mu}^e)$ 
and $\bf { 0 } $ are in the interior of the the convex cone generated by the
columns  of the matrix $\mathbf{\Gamma}^{e,\mathcal A}$.
\\

After the  $(k-1)$-th jump time, let $Y_{\tau_{k-1}}(\omega)= e $, and define
\begin{eqnarray}
\Xi_k(i,j) & := &\{ h=(n;e_0,\dots,e_n): n\ge k , e_{k-1}=i, e_k=j \}\subseteq \Xi
\\
\hat{\mathcal Y}^{(e)}_k(\omega) &:=&\biggl \{ f: \exists h\in
(\Xi_k(e,f)\setminus \mathcal{ D}_k) 
\mbox{ and }h^{(k-1)}= H^{(k-1)}(\omega)\biggr \}
\end{eqnarray}

Similarly, for a  $\mathbb{G}$- predictable time $\overline \tau_k(h)=\underline \tau_k(h)  ,h\in{\mathcal  D}_k$ we define 
\begin{eqnarray*} 
\check{\mathcal Y}^{(e)}_{k,h}   (\omega):  = \biggl \{ f: 
\exists h' \in \Xi_k(e,f) \cap 
\mathcal{ D}_k   \mbox{ with  } \overline \tau_k(h) = \underline \tau_k(h)= \underline \tau_k(h')   = \overline \tau_k(h')
\mbox{ and }  h^{(k-1)}=H^{(k-1)}(\omega) \biggr \}
\end{eqnarray*}

$ \hat{\mathcal Y}^{(Y_{\tau_{k-1}})}_k(\omega)  $ is the set of states reachable
by one $\mathbb {G}$-totally inaccessible jump after time $\tau_{k-1}$.
\\

$\check{\mathcal Y}^{(Y_{\tau_{k-1}})}_{k,h}$
is the set of states reachable after time $\tau_{k-1}$ by one accessible transition
at the  $\mathbb {G}$-predictable time  $ \overline \tau_k(h), h\in {\mathcal D}_k$.
\\

Note that these random sets are determined at time ${\tau_{k-1}}$ in the insider
filtration $\mathbb{G}$.
\\

\begin{theorem} \label{NFLVR}
NFLVR is equivalent to the following  condition:

\begin{enumerate} 
\item (Totally  inaccessible jump part  ):
$\mathbb{P}$-almost surely, for all $k$  the constrained linear system  (\ref{matrix system2}) 
with  ${\mathcal A}= \hat{\mathcal Y}^{(Y_{\tau_{k-1}})}_k(\omega)$  has strictly positive solutions.
\item (Accessible  jump part  ): $\mathbb{P}$-almost surely, for all $k$ and all
$\mathbb{G}$-predictable times  $ \overline \tau_k(h), h \in {\mathcal D}_k$,
the homogeneous constrained linear system (\ref{matrix system3})
with ${\mathcal A}=\check{\mathcal Y}^{(Y_{\tau_{k-1}})}_{k,h}(\omega)$ has strictly positive solutions.
\end{enumerate} 
\end{theorem}

\begin{proof}
Any choice of positive solutions of the linear systems (\ref{matrix system2}) and 
(\ref{matrix system3}), for all $k$ 
and all $h\in {\mathcal D}_k$ corresponds in the standard way to a $\mathbb{G}$-martingale 
measure $\mathbb{Q}^{L}$. 
\end{proof}

%\vspace{15pt}
\begin{corollary} 
Define
\begin{eqnarray*}  &&
\tau' := \min \{ \tau_{k-1}: k\ge 1,\mbox{ the  NFLVR-condition (1) fails } \} \\ &&
\tau'' := \inf \{ \overline\tau_k(h): k\ge 0, h\in {\mathcal D}_k, \mbox{ and  NFLVR-condition  (2) fails } \} \\ &&
\tau^{FLVR} := \bigl(\tau'\wedge \tau''\wedge T \bigr)
\end{eqnarray*} 
There are not arbitrage possibilities for a $\mathbb{G}$- insider who is
restricted to trade in the interval $[0,\tau^{FLVR})$,
equivalently, there is an equivalent $\mathbb{G}$- martingale measure
for the stopped process $(\widetilde S_{t\wedge \tau^{FLVR}})$.
\end{corollary} 

We describe an arbitrage strategy, for an insider which is allowed to
trade after the $\mathbb{G}$-stopping times  $\tau'$ or  $\tau''$,
when $P(\tau^{FLVR} < \infty)>0$.
\begin{enumerate}
 \item 
For simplicity assume that $\tau'(\omega)= \tau_{k-1}(\omega)$ and
 the next jump time  $\tau_k$  is $\mathbb{G}$-totally inaccessible.

Let $Y_{\tau_{k-1}}(\omega)=e$.
 Since the non-homogeneous system $6.1$ has not strictly positive
solutions, by the separating hyperplane theorem there is  a 
vector $\xi=(\xi_1,\dots,\xi_m)$ such that for
\begin{eqnarray} \label{separating:hyperplane} 
\sum_{i}  \sum_{f\in {\mathcal A} }  \xi_i \gamma^{i ef} \lambda^{ef} + \sum_{i}\mu ^{ie } \xi_i > 0
\end{eqnarray}
for all  vectors $( \lambda^{ef}: f \in {\mathcal A} ) > 0 $. This implies that
\begin{eqnarray*} 
 \sum_{i}\mu ^{ie } \xi_i  > 0, \quad\mbox{ and } \quad 
\sum_{i}  \xi_i \gamma^{i ef} > 0 \; \forall f\in {\mathcal A}
\end{eqnarray*}

In a discounted world,
at any time
$s\in [ \tau_{k-1}, \tau_k)$, the insider starts to  play and
borrows $V=(\xi\cdot S_{\tau'})$   from the 
bank at zero interest rate, to buy a portfolio of stocks $S$ 
with weights $(\xi_1,\dots,\xi_m)$. The insider sells its' portfolio at
any time $t\in (s, \tau_k]$ and pays back is debt to the bank.
 Whether $\{ \tau_k = t \}$
or $\{\tau_k > t\}$,  from condition
(\ref{separating:hyperplane}) we see that 
the insider makes  a positive profit.

\item
Next we discuss the insider's strategy at the time $\tau''$.
Assume that $\tau_{k-1}(\omega)< \tau''(\omega) \le \tau_k(\omega) $ and let $Y_{\tau_{k-1}}(\omega)=e$.

Since the homogeneous system 
(\ref{matrix system3}) has not strictly positive solutions, by the separating hyperplane theorem
there is a vector $\xi \in \mathbb R^m$ such that
\begin{eqnarray} \label{separating:hyperplane2} 
\sum_{i}  \sum_{f\in {\mathcal A} }  \xi_i \gamma^{i ef} \lambda^{ef} > 0
\end{eqnarray}
for all  vectors $( \lambda^{ef}: f \in {\mathcal A} ) > 0 $. This implies 
\begin{eqnarray*} 
\sum_{i}  \xi_i \gamma^{i ef} > 0 \quad  \forall f\in {\mathcal A} .
\end{eqnarray*}
Since $\tau''$ is ${\mathcal F}_{\tau_{k-1}}$-measurable, the insider chooses
some $\varepsilon >0 $ small enough so that
$\tau_{k-1}(\omega)<   \tau''(\omega) - \varepsilon <    \tau''(\omega) \le \tau_k(\omega) $,
buys the portfolio $(\xi \cdot S_s )$ at time $s=(\tau''-\varepsilon)$, 
and sells the portfolio at time $\tau''$ after the jump, making a profit
\begin{eqnarray*} 
 \sum_{i} \sum_{f} \xi_i \gamma^{i ef} \Delta N^{ef}_{\tau''}  - \varepsilon \sum_{i} \xi_i \mu^{i e} 
\end{eqnarray*}
Since $\varepsilon$ is arbitrarily small the insider has a free lunch with vanishing risk,
regardless of the sign of $(\xi\cdot \mu^{e})$.
\end{enumerate}

\begin{example}\label{ex:only1+2}
Consider first the Example \ref{ex:one}. Using the notation from Remark \ref{remark3} we must compute 
$$
\bar \nu (d\ell , dt) = \mathbb P(\vartheta \in d\ell , N(dt) = 1\vert F_{t-} ) =  
\mathbb P(\nu \in d\ell \vert F_{t-})\mathbb P (N(dt) = 1 \vert F_{t-}, \vartheta \in d \ell ) 
$$
and 
$$
\widetilde \nu (d\ell , dt ) = \mathbb P(\vartheta  \in d\ell \vert F_{t-})\Lambda (dt) = \mathbb P (\vartheta  \in d\ell \vert F_{t-}) \mathbb P(N(dt)=1 \vert F_{t-})  .
$$
Recall from Example \ref{ex:one} that in fact we have here two independent Poisson processes $N^+ $ (resp. $N^-$)
counting the positive (resp. negative) jumps. We have 
$\mathbb P (N^+(dt)= 1 \vert F_{t-}) = \lambda ^+ dt $ and $\mathbb P (N^-(dt)= 1 \vert F_{t-}) = \lambda ^- dt $.
To compute the conditional probability $\mathbb P (N^+(dt) = 1 \vert F_{t-}, \vartheta \in d \ell )$ recall that 
$\vartheta = L_0 + \mu T + \beta ^+ N_T^+ + \beta ^- N^-_T $. Assume that $\beta ^+ , \beta ^- $ 
are such that with fixed $\vartheta = \ell $ the equation 
\begin{equation}\label{eq:ell}
\ell = \log (S_0 ) + \mu T + \beta ^+ n^+ + \beta ^- n^- 
\end{equation}
has a unique solution $(n^+ , n^- )$. We have then 
that 
$$
\mathbb P (N^+(dt) = 1 \vert F_{t-}, \vartheta \in d \ell ) = \mathbb P (N^+(dt) = 1 \vert F_{t-}, N^+ _T = n^+).
$$
Recall that for a Poisson process $N$ the compensator of $N$ in the filtration $\mathbb F \wedge \sigma (N_T) $
is 
$$
\mathbb P (N(dt) = 1 \vert F_{t-}, N_T ) = \frac{N_T - N_{t-}}{T-t} dt 
$$
(see for example \cite{Gasbarra2006}).
This gives 
\begin{equation}\label{eq:compensator-KH}
\mathbb P (N^+(dt) = 1 \vert F_{t-}, \vartheta \in d \ell ) =  
\mathbb P (N^+(dt) = 1 \vert F_{t-}, N_T^+ ) = \frac{N^+_T - N^+ _{t-}}{T-t} dt .
\end{equation}

In this model there is always arbitrage, after the last jump of $N^+ $ or $N^-$. 

Note that in the special case of $\frac{\beta ^+}{-  \beta ^-} = \frac{k_1}{k_2} $ for some $k_1, k_2\in \mathbb N $,  the equation 
(\ref{eq:ell}) does not uniquely determine the pair $(n^+ , n^-) $, and then there might also be no-arbitrage.
We refer to \cite{Kohatsu-Higa2007} for a detailed analysis of this model using the expected utility of the insider. 
\end{example}

\section*{Acknowledgments}
We are grateful to an anonymous referee for a detailed and constructive reading of the first 
versions. I.M. was supported by Magnus Ehrnrooth foundation and FGSS, and I. M. and E.V. were
supported by the Academy of Finland grants 210465 and 127634.
%%For cite command type as \cite{1}; \cite{3,6} and \cite{2,4,6}. 
%%For refcite command type as Refs.~[\refcite{1}];   
%%[\refcite{1},\refcite{3}] and [\refcite{1}--\refcite{4}].

%\section*{References}

\end{document}